\documentclass[%
 aip,
 jmp,%
 amsmath,amssymb,
 preprint%
]{revtex4-1}

\usepackage{amsthm}
\newtheorem{theo}{Theorem}
\newtheorem{defi}{Definition}
\newtheorem{lemm}{Lemma}
\newtheorem{prop}{Proposition}

\usepackage{braket}
\usepackage{mathrsfs}

\usepackage{graphicx}% Include figure files
\usepackage{dcolumn}% Align table columns on decimal point
\usepackage{bm}% bold math
\begin{document}

\preprint{AIP/123-QED}

\title[]{Construction of the least informative observable 
conserved by a given quantum instrument}
%\thanks{Footnote to title of article.}

\author{Yui Kuramochi}
 \affiliation{Department of Nuclear Engineering, Kyoto University, 6158540 Kyoto, Japan}%Lines break automatically or can be forced with \\
 \email{kuramochi.yui.22c@st.kyoto-u.ac.jp}

\date{\today}

\begin{abstract}
For a quantum measurement process described by a quantum instrument $\mathcal{I}$
and a system observable corresponding to a positive-operator valued measure (POVM) $E ,$
$\mathcal{I}$ is said to conserve the information of $E$
if the joint successive measurement of $\mathcal{I}$ followed by $E$ is 
equivalent to a single measurement of $E .$
%In this paper we consider a problem whether there exists a conserved POVM $E$
%for a given quantum instrument $\mathcal{I}.$
We show that for any quantum instrument $\mathcal{I}$
we can construct a POVM conserved by $\mathcal{I}$.
Intuitively the construction gives the infinite joint successive measurement of $\mathcal{I}.$
We also show that the constructed POVM is the least informative observable among POVMs conserved by $\mathcal{I}$,
i.e. the constructed POVM can be realized by 
a classical post-processing of any POVM conserved by $\mathcal{I} .$
As typical examples of quantum instruments,
we explicitly evaluate POVMs of infinite successive measurements 
for photon counting and quantum counter instruments.

%For a given quantum instrument $\mathcal{I}$ and a positive-operator valued measure (POVM) $E$,
%$E$ is said to be conserved by $\mathcal{I}$ if a joint successive measurement of $\mathcal{I}$ followed by $E$ is 
%equivalent to a single measurement of $E$
%in the sense that each one of them is realized by a classical post-processing of the other.
%For a given instrument $\mathcal{I}$ with a standard Borel outcome space, 
%we construct a POVM called an infinite composition of $\mathcal{I}$
%corresponding to the infinite joint successive measurement of $\mathcal{I}$ 
%by using a quantum version of the Kolmogorov extension theorem.
%We show that the infinite composition is the least informative POVM among the POVMs conserved by $\mathcal{I}$.
%As physical examples of instruments, we consider photon counting and quantum counter instruments
%and explicitly evaluate their infinite compositions.
\end{abstract}

%Use showkeys class option if keyword
\maketitle

%%%%%%%%%%%%%%%%%%%%%%Introduction %%%%%%%%%%%%%%%%%%%%%%%%%%%%%%%
\section{Introduction}
While quantum measurement back-action causes
the state change and the information loss of the system,
some kind of information is known to be conserved 
if we focus on a proper system observable
described by a positive-operator valued measure.
An example of such measurement processes is the photon counting measurement,
also known as the Srinivas-Davies model~\cite{doi10.1080/713820643},
in which we can estimate photon number of the pre-measurement state
from that of the post-measurement state and the measurement outcome.
In general, a measurement process described by a quantum 
instrument~\cite{content/aip/journal/jmp/25/1/10.1063/1.526000}
$\mathcal{I}$
is said to conserve a system observable described by a positive-operator valued measure $E$
if the joint measurement of $\mathcal{I}$ followed by $E$ is equivalent to a single measurement of $E$.

%In the previous work~\cite{PhysRevA.91.032110},
%the author found a sufficient condition for the \lq\lq{}relative-entropy conservation law,\rq\rq{}
%for system observable $X$ and quantum measurement process $Y$
%which is the generalization of \lq\lq{}Shannon entropy conservation\rq\rq{} derived 
%by Ban~\cite{0305-4470-32-9-012}.
%The interpretation of the derived sufficient condition
%is that the joint successive measurement of $Y$ followed by $X$ is \lq\lq{}equivalent\rq\rq{}
%to the single measurement of $X$ in the sense of the classical post-precessing.
%This \lq\lq{}equivalence\rq\rq{} relation, 
%while the formulation of Ref.~\onlinecite{PhysRevA.91.032110}
%is largely based on the concept of the sufficient statistic~\cite{halmos1949application},
%can be formulated more rigorously and briefly 
%by using the \textit{fuzzy} preorder and equivalence relations among 
%POVMs~\cite{Dorofeev1997349,Heinonen200577,jencova2008}.
%The reformulation of the information conservation condition 
%based on this fuzzy equivalence relation is one purpose of the present paper.

In the previous work~\cite{PhysRevA.91.032110},
the author examined some physical examples of quantum instruments
and showed that some intuitively \lq\lq{}natural\rq\rq{} observable for each instrument 
satisfies the conservation condition.
Then, it is natural to ask whether there exists a POVM $E$ conserved by a given quantum instrument 
$\mathcal{I}$.
The answer, the main result of this paper,
is affirmative, and $E$ can be constructed as the infinite successive joint measurement
of $\mathcal{I}$, called an infinite composition of $\mathcal{I}$.
Furthermore, this infinite successive measurement is shown to be characterized by 
the minimality up to the fuzzy preorder relation~\cite{Dorofeev1997349,de2002foundations,Heinonen200577,jencova2008} 
among the conserved POVMs,
i.e. the infinite successive measurement is the least informative POVM conserved 
by $\mathcal{I}$.
We also reconsider the photon counting and quantum counter instruments 
and explicitly derive their infinite compositions.
As a by-product of the discussion, we also correct a mathematical insufficiency
in the proof of the existing work~\cite{PhysRevA.53.3808}
concerning the convergence of the normalized count number
in the quantum counter measurement.

In Ref.~\onlinecite{PhysRevA.91.032110},
the conservation condition was introduced as
a sufficient condition for the \lq\lq{}relative-entropy conservation law\rq\rq{}
for system observable $E,$ 
which is the generalization of \lq\lq{}Shannon entropy conservation\rq\rq{} derived 
by Ban~\cite{0305-4470-32-9-012}.
In the comparison of two POVMs, 
the discussion of Ref.~\onlinecite{PhysRevA.91.032110}
is largely based on the concept of the sufficient statistic~\cite{halmos1949application},
while it can be formulated more rigorously and briefly 
by using the \textit{fuzzy} preorder and equivalence relations between
POVMs~\cite{Dorofeev1997349,de2002foundations,Heinonen200577,jencova2008}
(Definition~\ref{defi:cons}).
The reformulation of the information conservation condition 
based on the fuzzy equivalence relation is another purpose of the present paper.

This paper is organized as follows.
In Sec.~\ref{sec:2}, some preliminary results concerning the fuzzy preorder and equivalence relations
and the composition of quantum measurement processes 
are reviewed.
In Sec.~\ref{sec:3}, we construct the infinite composition of a given quantum instrument with 
a standard Borel outcome space and show that
it is the least informative POVM that is conserved by the instrument.
In Sec.~\ref{sec:4}, we consider photon counting and quantum counter instruments 
and derive the explicit forms of their infinite compositions.
Sec.~\ref{sec:5} summarizes the main results of this paper.

\section{Preliminaries}
\label{sec:2}
In this section we briefly review some preliminary results on the quantum theory of measurement and fix the notation.
For a general reference of quantum measurement, we refer
Refs.~\onlinecite{heinosaari2011mathematical,davies1976quantum,de2002foundations}.

\subsection{Positive-operator valued measures and fuzzy preorder and eqivalence relations}
We fix a complex Hilbert space $\mathcal{H}$
and denote the set of bounded linear operators on $\mathcal{H}$ as $\mathcal{L}(\mathcal{H})$.
Let $(\Omega , \mathscr{B})$ be a measurable space.
A mapping $E : \mathscr{B} \rightarrow \mathcal{L} (\mathcal{H})$
is called a \textit{positive-operator valued measure (POVM)} with its outcome space $(\Omega , \mathscr{B})$
if 
\begin{enumerate}
\item[(i)]
$E(B) \geq O$ ($\forall B \in \mathscr{B}$);
\item[(ii)]
$E(\Omega) = I$;
\item[(iii)]
for any disjoint $\{ B_i \}_{i \geq 1} \subset \mathscr{B}$,
$E(\cup_{i\geq 1}  B_i ) = \sum_{i\geq 1} E(B_i)$
in the weak operator topology.
\end{enumerate}
Here $O$ and $I$ are zero and identity operators, respectively.
Let $E$ be a POVM with its outcome space $(\Omega , \mathscr{B})$.
A measurable set $N \in \mathscr{B}$ is called an $E$-null set
if $E(N)= O.$
$E$-almost sure equations and convergences for stochastic variables are also defined in a similar manner 
as for a classical probability measure.

Let 
$(\Omega_1 , \mathscr{B}_1)$
and
$(\Omega_2 , \mathscr{B}_2)$
be measurable spaces.
A mapping
$\nu_\cdot (\cdot) \colon \Omega_2 \times \mathscr{B}_1
\ni (\omega_2 , B)
\mapsto 
\nu_{\omega_2} (B)
\in
[0,1]$
is said to be a \textit{Markov kernel} if
\begin{enumerate}
\item[(i)]
for each $\omega_2 \in \Omega_2$, the mapping
$\nu_{\omega_2 } (\cdot) \colon \mathscr{B}_1 \to [0,1]$
is a probability measure;
\item[(ii)]
for each $B \in \mathscr{B}_1$, the mapping
$\nu_\cdot (B) \colon \Omega_2 \to [0,1]$ is $\mathscr{B}_2$-measurable.
\end{enumerate}
Let $E^1$ and $E^2$ be POVMs with their outcome spaces 
$(\Omega_1 , \mathscr{B}_1)$
and 
$(\Omega_2 , \mathscr{B}_2)$,
respectively.
Following Refs.~\onlinecite{de2002foundations,Heinonen200577,jencova2008}, 
we define a relation $E^1 \preceq  E^2$ by the existence of a Markov kernel
$\nu_\cdot (\cdot)  \colon \Omega_2 \times \mathscr{B}_1 \to [0,1]$
such that
\begin{equation}
	E^1 (B)
	=
	\int_{\Omega_2} 
	\nu_{\omega_2} (B) 
	E^2 (d\omega_2) 
	\quad
	(\forall B \in \mathscr{B}_1),
	\label{eq:fuzzy}
\end{equation}
and we say that $E^1$ is \textit{fuzzier} than $E^2$.
POVMs $E^1$ and $E^2$ are said to be \textit{equivalent},
denoted as $E^1 \simeq E^2 ,$
if $E^1 \preceq E^2$ and $E^2 \preceq E^1 .$
The relations
$\preceq$ and $\simeq$ are preorder and equivalence relations
for POVMs~\cite{Heinonen200577,jencova2008}, respectively.
Equation~\eqref{eq:fuzzy} intuitively means that
the measurement of $E^1$ can be realized by performing the measurement of $E^2$
and the classical information processing on the measurement outcome $\omega_2.$

The following lemma concerning the Markov kernel will be used later.
%%%%%%%%%%%%%%%%%%%%% lemma 1%%%%%%%%%%%%%%%%%%%%%%
\begin{lemm}
\label{lemm:kernel}
\begin{enumerate}
\item
Let $(\Omega_i , \mathscr{B}_i)$ $(i = 1,2,3)$ be measurable spaces. 
Suppose that
\begin{gather*}
	\nu^{1}_{\cdot} (\cdot) 
	\colon 
	\Omega_2 \times \mathscr{B}_1
	\to
	[0,1]
	\\
	\nu^{2}_{\cdot} (\cdot) 
	\colon 
	\Omega_3 \times \mathscr{B}_2
	\to
	[0,1]
\end{gather*}
are Markov kernels.
Then $\nu^3_{\cdot} (\cdot)$ defined by
\[
	\nu^3_{\omega_3} (B)
	:=
	\int_{\Omega_2}
	\nu^1_{\omega_2}(B)
	\nu^2_{\omega_3} (d\omega_2)
	\quad
	(B \in \mathscr{B}_1 , \omega_3 \in \Omega_3)
\]
is a Markov kernel.
\item
Let 
$(\Omega_i , \mathscr{B}_i)$
$(i =1,2,3)$ 
be measurable spaces and let
$
	\nu_{\cdot} (\cdot) 
	\colon 
	\Omega_3 \times \mathscr{B}_2
	\to
	[0,1]
$
be a Markov kernel.
Define $\tilde{\nu}_\cdot (\cdot)$ by
\begin{gather}
	\tilde{\nu}_{(\omega_1, \omega_3)} (B)
	:=
	\nu_{\omega_3} ( \left. B \right\rvert_{\omega_1}  ) ,
	\label{eq:lemm1}
	\\
	B \rvert_{\omega_1}
	:=
	\{
	\omega_2 \in \Omega_2 |
	(\omega_1 , \omega_2)
	\in 
	B
	\}
	\notag
\end{gather}
for $(\omega_1 ,\omega_3) \in \Omega_1 \times \Omega_3 $
and $B \in \mathscr{B}_1 \times \mathscr{B}_2$.
Then $\tilde{\nu}_\cdot (\cdot )$ is a Markov kernel.
\end{enumerate}
\end{lemm}
%%%%%%%%%%proof of lemma 1%%%%%%%%
\begin{proof}
1 is shown in the proof of the Proposition~1 in Ref.~\onlinecite{Heinonen200577}.
We show 2. 
For each $(\omega_1 ,\omega_3) \in \Omega_1 \times \Omega_3 $,
$\tilde{\nu}_{(\omega_1 ,\omega_3)} (\cdot ) $ is a probability measure.
To show the measurability of $\tilde{\nu}_\cdot (B)$ ($B \in \mathscr{B}_1 \times \mathscr{B}_2 $),
define a class $\mathscr{D}$ of subsets of $\Omega_1 \times \Omega_2$ by
\[
	\mathscr{D}
	:=
	\Set{
	B
	\in \mathscr{B}_1 \times \mathscr{B}_2 
	|
	\text{
	$
	\tilde{\nu}_{\cdot} (B)
	$
	is $\mathscr{B}_1 \times \mathscr{B}_3$-measurable}
	}.
\]
Then $\mathscr{D}$ is a Dynkin system ($\lambda$-system),
i.e. it is closed under countable disjoint unions and proper differences and contains $\Omega_1 \times \Omega_2$. 
For each
$B_i \in \mathscr{B}_i$ $(i=1,2)$,
\[
\tilde{\nu}_{(\omega_1, \omega_3)} (B_1 \times B_2 ) 
= 
\chi_{B_1} (\omega_1) 
\nu_{\omega_3} (B_2)
\]
is $\mathscr{B}_1 \times \mathscr{B}_3$-measurable,
where $\chi_B (\cdot)$ is an indicator function for a subset $B$.
Thus $\mathscr{D}$ contains the class of the cylinder sets
$\Set{ B_1 \times B_2 |  B_1 \in \mathscr{B}_1 , B_2 \in \mathscr{B}_2 }$
and the Dynkin's theorem assures that $\mathscr{D}$ coincides with $\mathscr{B}_1 \times \mathscr{B}_2$,
which proves the assertion.
\end{proof}

\subsection{Completely positive instruments and their compositions}
Let $(\Omega , \mathscr{B})$ is a measurable space.
A \textit{completely positive (CP) instrument}~\cite{davieslewisBF01647093,davies1976quantum,%
content/aip/journal/jmp/25/1/10.1063/1.526000} 
(in the Heisenberg picture)
with its outcome space $(\Omega , \mathscr{B})$
is a mapping 
\[
\mathcal{I}_{\cdot} (\cdot)
\colon
\mathscr{B} \times \mathcal{L}(\mathcal{H})
\ni
(B , a)
\mapsto
\mathcal{I}_{B} (a)
\in 
\mathcal{L} (\mathcal{H})
\]
such that
\begin{enumerate}
\item[(i)]
for any $B \in \mathscr{B}$,
$\mathcal{I}_B (\cdot ) \colon \mathcal{L} (\mathcal{H}) \to \mathcal{L} (\mathcal{H})$
is a normal CP linear map;
\item[(ii)]
for any disjoint $\{ B_i \}_{i \geq 1} \subset \mathscr{B}$
and any $a \in \mathcal{L} (\mathcal{H})$,
$
\mathcal{I}_{\cup_{i\geq 1} B_i}  (a) 
= \sum_{i\geq 1}
\mathcal{I}_{B_i} (a)
$
in the ultraweak operator topology;
\item[(iii)]
$\mathcal{I}_\Omega (I) = I .$
\end{enumerate}
We also define a positive (P) instrument by replacing CP with P in the definition of the CP instrument.
A CP instrument $\mathcal{I}_{\cdot} (\cdot)$ with its outcome space $(\Omega , \mathscr{B})$
describes the statistics of the measurement outcome and the state change due to the measurement
simultaneously in a necessary and sufficient manner.
The POVM corresponding to the probability distribution of the measurement outcome is given by
$E(B) = \mathcal{I}_B (I)$ $(B \in \mathscr{B})$.
Here we only consider the case when
the range and the domain of $\mathcal{I}_B (\cdot)$ is identical,
which is a necessary condition so that we can compare the \textit{same} observable
of the system before and after the measurement.

A measurable space $(\Omega , \mathscr{B})$ Borel isomorphic to a complete separable metric space
is called a \textit{standard Borel} space~\cite{srivastava1998course}.
A (C)P instrument (resp. POVM) with a standard Borel outcome space is called 
a standard Borel (C)P instrument (resp. a standard Borel POVM).
\textit{In the rest of this paper we only consider standard Borel (C)P instruments and POVMs.}

The following theorem,
which is a slight modification of the theorem
due to Davies and Lewis~\cite{davieslewisBF01647093,davies1976quantum}, 
assures the existence of a CP instrument and a POVM corresponding to 
a joint successive measurement process.
\begin{theo}
\label{theo:comp}
\begin{enumerate}
\item[(i)]
Let $\mathcal{I}^i_\cdot (\cdot)$ be a CP instrument with 
a standard Borel outcome space $(\Omega_i , \mathscr{B}_i)$ $(i = 1,2)$.
Then there exists a unique CP instrument $\mathcal{I}^{12}_\cdot (\cdot)$
with the product outcome space $(\Omega_1 \times \Omega_2 , \mathscr{B}_1 \times \mathscr{B}_2)$
such that
\begin{equation}
	\mathcal{I}^{12}_{B_1 \times B_2} (\cdot)
	=
	\mathcal{I}^1_{B_1} 
	\circ
	\mathcal{I}^2_{B_2} (\cdot) 
	\quad
	(B_1 \in \mathscr{B}_1,  B_2 \in \mathscr{B}_2  ).
	\label{eq:th1}
\end{equation}
\item[(ii)]
Let $\mathcal{I}^1_\cdot (\cdot)$ be a CP instrument with 
a standard Borel outcome space $(\Omega_1 , \mathscr{B}_1)$
and let $E^2$ be a POVM with a standard Borel outcome space $(\Omega_2, \mathscr{B}_2)$.
Then there exists a unique POVM $E^{12}$ with the product outcome space 
$(\Omega_1 \times \Omega_2 , \mathscr{B}_1 \times \mathscr{B}_2)$
such that
\begin{equation}
	E^{12} (B_1 \times B_2)
	=
	\mathcal{I}^1_{B_1} ( E^2(B_2))
	\quad
	(B_1 \in \mathscr{B}_1,  B_2 \in \mathscr{B}_2  ).
	\label{eq:th12}
\end{equation}
\end{enumerate}
\end{theo}
We call the CP instrument $\mathcal{I}^{12}_\cdot (\cdot)$ 
and the POVM $E^{12}$ as \textit{compositions}, and denote them as
$(\mathcal{I}^1 \ast \mathcal{I}^2)_\cdot (\cdot)$ and $\mathcal{I}^1 \ast E^2$,
respectively.
\begin{proof}
We first show (i).
According to Theorem~4.2.2 of Ref.~\onlinecite{davies1976quantum},
there exists a unique P instrument $\mathcal{I}^{12}_\cdot (\cdot)$ such that
the condition~\eqref{eq:th1} holds.
To show the complete positivity of $\mathcal{I}^{12}_\cdot (\cdot)$,
define a class $\mathscr{D}$ of subsets of $\Omega_1 \times \Omega_2$ by
\[
	\mathscr{D}
	:=
	\Set{
	B \in \mathscr{B}_1 \times \mathscr{B}_2
	|
	\text{
	$\mathcal{I}^{12}_B (\cdot )$
	is CP
	}
	} .
\]
From the condition~\eqref{eq:th1}, $\mathscr{D}$ contains the class of cylinder sets.
Since we can easily verify that $\mathscr{D}$ is a Dynkin class,
Dynkin's theorem assures that $\mathscr{D} = \mathscr{B}_1 \times \mathscr{B}_2$,
proving the assertion (i).

To show (ii),
take a CP instrument $\mathcal{I}^2_\cdot (\cdot)$ with the outcome space 
$(\Omega_2, \mathscr{B}_2)$
such that 
$\mathcal{I}^2_B (I) = E^2 (B)$ $(B \in \mathscr{B}_2)$.
Then a POVM $E^{12} (\cdot) := (\mathcal{I}^1 \ast \mathcal{I}^2)_\cdot (I)$
satisfies the condition~\eqref{eq:th12}.
The uniqueness can be shown by using the Dynkin's theorem
as parallel as in the classical measure.
\end{proof}

Now let $\mathcal{I}^k_\cdot (\cdot)$ be a CP instrument
with a standard Borel outcome space $(\Omega_k , \mathscr{B}_k)$
$(k = 1,2, \cdots)$
and let $E$ be a POVM with a standard Borel outcome space $(\Omega , \mathscr{B})$.
Then we have the following associative laws for the composition:
\begin{gather*}
	(
	\mathcal{I}^1
	\ast
	\mathcal{I}^2
	)
	\ast 
	\mathcal{I}^3
	=
	\mathcal{I}^1
	\ast
	(
	\mathcal{I}^2
	\ast
	\mathcal{I}^3
	),
	\\
	(
	\mathcal{I}^1
	\ast
	\mathcal{I}^2
	)
	\ast 
	E
	=
	\mathcal{I}^1
	\ast
	(
	\mathcal{I}^2
	\ast
	E
	).
\end{gather*}
Thus we may write these CP instrument and POVM as
$
	\mathcal{I}^1
	\ast
	\mathcal{I}^2
	\ast 
	\mathcal{I}^3
$
and
$
	\mathcal{I}^1
	\ast
	\mathcal{I}^2
	\ast
	E ,
$
respectively.
Multiple compositions
$
	\mathcal{I}^1
	\ast 
	\cdots
	\ast
	\mathcal{I}^n
$
and
$
	\mathcal{I}^1
	\ast 
	\cdots
	\ast
	\mathcal{I}^n
	\ast
	E
$
for general $n \geq 1$
are also defined in a similar manner.
These are the unique CP instrument and POVM
such that
\begin{gather*}
	(
	\mathcal{I}^1
	\ast 
	\cdots
	\ast
	\mathcal{I}^n
	)_{B_1 \times \cdots \times B_n}
	(\cdot)
	=
	\mathcal{I}^1_{B_1}
	\circ \cdots \circ
	\mathcal{I}^n_{B_n} 
	(\cdot),
	\\
	(
	\mathcal{I}^1
	\ast 
	\cdots
	\ast
	\mathcal{I}^n
	\ast
	E
	)
	(B_1 \times \cdots \times B_n \times B)
	=
	\mathcal{I}^1_{B_1}
	\circ \cdots \circ
	\mathcal{I}^n_{B_n} 
	(E(B))
\end{gather*}
for each $B_k \in \mathscr{B}_k$ $(k=1, \cdots , n)$
and $B \in \mathscr{B}$.

For later use, we show the following lemmas.
%%%%%%%%%%%%lemma koukan%%%%%%%%%%%%%%%%%%%%
\begin{lemm}
\label{lemm:koukan}
Let $\mathcal{I}^1_\cdot (\cdot)$ be a CP instrument with 
a standard Borel outcome space $(\Omega_1 , \mathscr{B}_1)$
and let $E^2$ be a POVM with a standard Borel outcome space $(\Omega_2, \mathscr{B}_2)$.
Then for 
each $B_1 \in \mathscr{B}_1$ and 
for each bounded complex valued $\mathscr{B}_2$-measurable function $f$,
\begin{equation}
	\int_{\Omega_1 \times \Omega_2}
	\chi_{B_1} (\omega_1)
	f(\omega_2) 
	(\mathcal{I}^1 \ast E^2 )
	(d \omega_1 \times d \omega_2)
	=
	\mathcal{I}^1_{B_1}
	\left(
	\int_{\Omega_2}
	f(\omega_2)
	E^2 (d \omega_2)
	\right) .
	\label{eq:lem2}
\end{equation}
\end{lemm}
\begin{proof}
It is sufficient to show Eq.~\eqref{eq:lem2}
when $f \geq 0.$
When $f$ is a measurable simple function,
Eq.~\eqref{eq:lem2} holds.
For general $f$, take a monotone sequence of non-negative measurable simple functions $f_n$
such that $f_n (\omega_2) \uparrow f(\omega_2)$
for each $\omega_2 \in \Omega_2$.
Then from the dominated convergence theorem we have
\begin{gather}
	\int_{\Omega_1 \times \Omega_2}
	\chi_{B_1} (\omega_1)
	f_n(\omega_2) 
	(\mathcal{I}^1 \ast E^2 )
	(d \omega_1 \times d \omega_2)
	\uparrow
	\int_{\Omega_1 \times \Omega_2}
	\chi_{B_1} (\omega_1)
	f(\omega_2) 
	(\mathcal{I}^1 \ast E^2 )
	(d \omega_1 \times d \omega_2) ,
	\label{eq:lem2-1}
	\\
	\int_{\Omega_2}
	f_n(\omega_2)
	E^2 (d \omega_2)
	\uparrow
	\int_{\Omega_2}
	f(\omega_2)
	E^2 (d \omega_2) .
	\label{eq:lem2-2}
\end{gather}
Since $\mathcal{I}^1_{B_1}(\cdot)$ is normal,
Eq.~\eqref{eq:lem2-2} implies that
\begin{equation}
	\mathcal{I}^1_{B_1}
	\left(
	\int_{\Omega_2}
	f_n(\omega_2)
	E^2 (d \omega_2)
	\right)
	\uparrow
	\mathcal{I}^1_{B_1}
	\left(
	\int_{\Omega_2}
	f(\omega_2)
	E^2 (d \omega_2)
	\right).
	\label{eq:lem2-3}
\end{equation}
Since LHSs of Eqs.~\eqref{eq:lem2-1} and \eqref{eq:lem2-3} coincide,
we obtain Eq.~\eqref{eq:lem2}.
\end{proof}
%%%%%%%%%%%%%%%%%lemma relation%%%%%%%%%%%%%%%
\begin{lemm}
\label{lemm:relation}
Let $\mathcal{I}^1_\cdot (\cdot)$ be a CP instrument with 
a standard Borel outcome space $(\Omega_1 , \mathscr{B}_1)$
and let $E^2$ and $E^3$ be POVMs with  standard Borel outcome spaces 
$(\Omega_2, \mathscr{B}_2)$
and 
$(\Omega_3, \mathscr{B}_3)$,
respectively.
Then we have
\begin{enumerate}
\item[(i)]
if $E^2 \preceq E^3$,
then
$\mathcal{I}^1 \ast E^2 \preceq  \mathcal{I}^1 \ast E^3 ;$
\item[(ii)]
if $E^2 \simeq E^3$,
then
$\mathcal{I}^1 \ast E^2 \simeq  \mathcal{I}^1 \ast E^3 .$
\end{enumerate}
\end{lemm}
\begin{proof}
We first show (i).
From the assumption $E^2 \preceq E^3$
there exists a Markov kernel 
$
\nu_\cdot (\cdot) 
\colon \Omega_3 \times \mathscr{B}_2
\to
[0,1]
$
such that
\begin{equation*}
	E^2 (B_2)
	=
	\int_{\Omega_3}
	\nu_{\omega_3} (B_2)
	E^3 (d\omega_3)
	\quad
	(B_2 \in \mathscr{B}_2).
\end{equation*}
From Lemma~\ref{lemm:kernel},
we can define a POVM $E^{12}$ 
by
\[
	E^{12} (B)
	:=
	\int_{\Omega_1 \times \Omega_3}
	\nu_{\omega_3}
	(B \rvert_{\omega_1})
	(\mathcal{I}^1 \ast E^3) (d\omega_1\times d \omega_3 )
	\quad
	(B \in \mathscr{B}_1 \times \mathscr{B}_2).
\]
From the definition of $E^{12}$, 
$E^{12} \preceq \mathcal{I}^1 \ast E^3$ holds.
On the other hand,
for each 
$B_1 \in \mathscr{B}_1$ 
and
$B_2 \in \mathscr{B}_2$,
we have
\begin{align*}
	E^{12} (B_1 \times B_2)
	&=
	\int_{\Omega_1 \times \Omega_3}
	\chi_{B_1} (\omega_1)
	\nu_{\omega_3}
	(B_2)
	(\mathcal{I}^1 \ast E^3) (d\omega_1\times d \omega_3 )
	\\
	&=
	\mathcal{I}^1_{B_1}
	\left(
	\int_{\Omega_3}
	\nu_{\omega_3} (B_2)
	E^3(d\omega_3)
	\right)
	\\
	&= 
	\mathcal{I}^1_{B_1} (E^2 (B_2))
	=
	(\mathcal{I}^1 \ast E^2) (B_1 \times B_2),
\end{align*}
where 
we have used Lemma~\ref{lemm:koukan}
in the derivation of the second equality.
Therefore we obtain
$\mathcal{I}^1 \ast E^2 = E^{12} \preceq \mathcal{I}^1 \ast E^3$
and (i) is proved.
(ii) immediately follows from (i) and the definition of 
the equivalence relation $\simeq$.
\end{proof}

\section{Infinite composition of an instrument and its minimal information-conserving property}
\label{sec:3}
\begin{defi}
\label{defi:cons}
Let $\mathcal{I}_\cdot (\cdot)$ be a standard Borel CP instrument
and let $E $ be a standard Borel POVM.
We say that
$E$ is conserved by $\mathcal{I}$,
or $\mathcal{I}$ conserves $E$,
if $\mathcal{I} \ast E \simeq E$.
\end{defi}
This condition is essentially the same as the one 
obtained in the author\rq{}s previous work~\cite{PhysRevA.91.032110}
for a sufficient condition for the \lq\lq{relative-entropy conservation law}\rq\rq{}
which is a generalization of the \lq\lq{Shannon-entropy conservation law}\rq\rq derived by
Ban~\cite{0305-4470-32-9-012}.

From Definition~\ref{defi:cons} and Lemma~\ref{lemm:relation},
we immediately obtain the following theorem.
%%%%%%%%%%%%%%Theorem independence as to equivalence %%%%%%%%%%%
\begin{theo}
\label{theo:independence}
Let $\mathcal{I}_\cdot (\cdot)$ be a standard Borel CP instrument
and let $E^{1} $ and $E^{2} $ be standard Borel POVMs.
Suppose that $E^{1} \simeq E^{2}.$
Then $\mathcal{I}_\cdot (\cdot)$ conserves $E^{1}$
if and only if
$\mathcal{I}_\cdot (\cdot)$ conserves $E^{2}.$
In other words, the conservation by $\mathcal{I}_\cdot (\cdot)$
is well-defined to $\simeq$-equivalence classes of standard Borel POVMs.
\end{theo}

Now we ask whether there exists a standard Borel POVM
$E$
that is conserved by a given standard Borel CP instrument $\mathcal{I}_\cdot (\cdot)$.
The answer is affirmative, and it is given by a POVM called the
\textit{infinite composition} of $\mathcal{I}$
corresponding to an infinite successive measurement of $\mathcal{I}$.
The conservation of the infinite composition $E_\infty$ by a CP instrument $\mathcal{I} $
is intuitively understood as follows:
the composition $\mathcal{I} \ast E_\infty$ is a measurement process 
in which we first perform $\mathcal{I}$ and then perform $\mathcal{I}$ infinitely many times,
which is obviously equivalent to performing $E_\infty .$
Furthermore,
we will show in Theorem~\ref{theo:minimal} that
the infinite composition $E_\infty$ of $\mathcal{I}$ 
is special to $\mathcal{I}$ in the sense that
$E_\infty$ is the minimal element with respect to the preorder relation $\preceq$
among POVMs conserved by $\mathcal{I}.$

The following proposition due to Tumulka~\cite{s11005-008-0229-8}
is a key to the construction of the infinite composition, 
which is a POVM version of the celebrated Kolmogorov extension theorem for probability 
measures~\cite{kolmogorov1933}.

\begin{prop}[quantum Kolmogorov extension theorem]
\label{prop:kolmogorov}
Let $(\Omega_i , \mathscr{B}_i)$ $(i=1,2, \cdots)$
be a standard Borel space
and let $E_n$ $(n=1,2, \cdots )$ be a POVM
with the product outcome space 
$(\prod_{i=1}^n \Omega_i ,  \prod_{i=1}^n \mathscr{B}_i )$.
Suppose that $\{  E_n \}$
satisfies the condition
\begin{equation}
	E_n (B)
	=
	E_{n+1} (B \times \Omega_{n+1})
	\quad 
	\left(
	n\geq 1, 
	B \in \prod_{i=1}^n \mathscr{B}_i
	\right).
	\label{eq:kolmogorov1}
\end{equation}
Then there exists a unique POVM $E_\infty$ with the infinite product outcome space 
$(\prod_{i=1}^\infty \Omega_i ,  \prod_{i=1}^\infty \mathscr{B}_i )$
such that
\begin{equation}
	E_n (B)
	=
	E_{\infty} 
	\left(
	B \times 
	\prod_{i=n+1}^\infty\Omega_{i}
	\right)
	\quad 
	\left(
	n\geq 1, 
	B \in \prod_{i=1}^n \mathscr{B}_i
	\right).
	\label{eq:kolmogorov2}
\end{equation}
\end{prop}
The condition~\eqref{eq:kolmogorov1} is called a \textit{Kolmogorov consistency condition}
and the POVM $E_\infty$
satisfying \eqref{eq:kolmogorov2} is said to be \textit{consistent} with
$\{  E_n \}$.
Note that the consistent POVM $E_\infty$
is a standard Borel POVM since
a countable product of standard Borel spaces is also a standard Borel space.

For simplicity, if $(\Omega_i , \mathscr{B}_i)$ $(1 \leq i \leq n)$
is identical to $(\Omega, \mathscr{B})$,
the product space 
$(\prod_{i=1}^n \Omega_i ,  \prod_{i=1}^n \mathscr{B}_i )$
$(1 \leq n \leq \infty)$
is denoted as $( \Omega^n , \mathscr{B}^n )$.
We also denote $\underbrace{ \mathcal{I} \ast \cdots \ast \mathcal{I} }_{  \text{$n$ elements}  }$
as $\mathcal{I}^{\ast n}$ for a CP instrument $\mathcal{I}$.

Now we construct the infinite composition of a CP instrument.

\begin{theo}[infinite composition of an instrument]
\label{theo:infinite}
Let $\mathcal{I}_\cdot (\cdot )$ be a CP instrument 
with a standard Borel outcome space $(\Omega, \mathscr{B})$.
Then there exists a unique POVM 
$E_\infty$
with the infinite product outcome space 
$(\Omega^\infty, \mathscr{B}^\infty)$
such that
\begin{equation}
	E_\infty
	\left(
	\prod_{i=1}^n B_i \times \Omega^\infty
	\right)
	=
	\mathcal{I}_{B_1}
	\circ \cdots \circ
	\mathcal{I}_{B_n}
	(I)
	\label{eq:infinite}
\end{equation}
for each $n \geq 1$ and $B_i \in \mathscr{B}$ 
$(i = 1,\cdots , n).$
The POVM $E_\infty$ is called an infinite composition of $\mathcal{I} .$
\end{theo}
\begin{proof}
Let us define a POVM $E_n (\cdot) := ( \mathcal{I}^{ \ast n})_{\cdot} (I)$
for each $n \geq 1.$
Since $E_n (\prod^n_{i=1} B_i)$ coincides with the RHS of Eq.~\eqref{eq:infinite},
from the quantum Kolmogorov extension theorem,
it is sufficient to show that $\{ E_n \} $
satisfies the Kolmogorov consistency condition.
For each $B \in \mathscr{B}^n$ $(n\geq 1)$,
we have
\begin{align*}
	E_{n+1} (B \times \Omega)
	&=
	(
	\mathcal{I}^{ \ast n} 
	\ast
	\mathcal{I}
	)_{B \times \Omega}
	(I)
	\\
	&=
	\mathcal{I}^{\ast n}_B (\mathcal{I}_{\Omega} (I))
	\\
	&=
	\mathcal{I}^{\ast n}_B (I)
	=E_n (B),
\end{align*}
and the theorem holds.
\end{proof}
The next theorem is the main result of this paper,
which states that for a given CP instrument $\mathcal{I},$ 
the infinite composition of $\mathcal{I}$ is the least informative
POVM conserved by $\mathcal{I} .$
%%%%%%%%%%%%%%%%main theorem%%%%%%%%%%%%%%
\begin{theo}
\label{theo:minimal}
Let $\mathcal{I}_\cdot (\cdot)$ is a CP instrument with a standard Borel outcome space
$(\Omega, \mathscr{B})$
and let $E_\infty$ be the infinite composition of $\mathcal{I}$.
Then $\mathcal{I}$ conserves $E_\infty$.
Furthermore $E_\infty$ is the minimal element with respect to the preorder relation
$\preceq$
among the standard Borel POVMs conserved by $\mathcal{I}$.
\end{theo}
\begin{proof}
%We denote $\prod_{i=1}^n y_i$ as $\bm{y}^{(\infty)}$
%for $1 \leq n \leq \infty .$
$(\Omega^\infty , \mathscr{B}^\infty)$ and 
$(\Omega \times \Omega^\infty , \mathscr{B} \times \mathscr{B}^\infty)$
are Borel isomorphic by the mapping
\[
	\Omega \times \Omega^\infty
	\ni
	\left(\omega , \prod_{i=1}^\infty \omega_i \right)
	\mapsto
	(\omega , \omega_1 , \omega_2 , \cdots)
	\in
	\Omega^\infty,
\]
and we consider $\mathcal{I} \ast E_\infty$ as a POVM with the outcome space
$(\Omega^\infty , \mathscr{B}^\infty)$
by this identification.
Then for each $B , B_1 , \cdots , B_n \in \mathscr{B}$
($1 \leq  n < \infty$),
we have
\begin{align*}
	(
	\mathcal{I} \ast E_\infty
	)
	(B\times B_1 \times \cdots B_n \times \Omega^\infty)
	&=
	\mathcal{I}_B
	(  E_\infty (  B_1 \times \cdots B_n \times \Omega^\infty  )  )
	\\
	&=
	\mathcal{I}_B
	\circ
	\mathcal{I}_{B_1}
	\circ
	\cdots 
	\circ
	\mathcal{I}_{B_n} (I)
	\\
	&=
	E_\infty
	( B \times B_1 \times \cdots B_n \times \Omega^\infty).
\end{align*}
From the uniqueness of $E_\infty$, we obtain $\mathcal{I} \ast E_\infty =  E_\infty$,
which proves the conservation of $E_\infty$ by $\mathcal{I}$.

To show the minimality of $E_\infty$,
take an arbitrary POVM $F$ with a standard Borel outcome space 
$( \Omega_X , \mathscr{B}_X )$ such that $F \simeq \mathcal{I} \ast F .$
The goal of the proof is to construct a Markov kernel corresponding to 
an information-procession from $\Omega_X$ to $\Omega^\infty $.
Since $  \mathcal{I} \ast F  \preceq   F$,
there exists a Markov kernel
$
\tilde{\nu}^1_\cdot (\cdot)
\colon
\Omega_X \times (\mathscr{B} \times \mathscr{B}_X)
\to
[0,1]
$
such that
\begin{equation}
	( \mathcal{I} \ast F ) (B)
	=
	\int_{\Omega_X}
	\tilde{\nu}^1_{x} (B)
	F (dx) 
	\notag
\end{equation}
for each $B \in \mathscr{B} \times \mathscr{B}_X .$
Since the Markov kernel $\tilde{\nu}_{x}^1 (d\omega_1 \times d x_1)$ 
corresponds to a classical information processing generating
a measurement outcome $(\omega_1, x_1)$ 
of $\mathcal{I} \ast F$
from a given measurement outcome $x$ of $F$,
we can construct a new Markov kernel $\tilde{\nu}^2_x (d \omega_1 \times d \omega_2 \times dx_2)$
which generates the measurement outcome of $\mathcal{I} \ast \mathcal{I} \ast F$ 
by applying the same classical information processing 
$\tilde{\nu}^1_{x_1}  (d \omega_2 \times dx_2)$
to $x_1$ which generates $(\omega_2 , x_2).$ 
Repeating the same discussion, we can construct a sequence 
of Markov kernels
$\tilde{\nu}^n_x (d\omega_1 \times \cdots \times d \omega_n \times dx_n )$
that generates the measurement outcome of $\mathcal{I}^{\ast n} \ast F$
from that of $F.$
The formal definition of
the sequence
$\{ \tilde{\nu}^n_\cdot (\cdot) \}$
is given by
\begin{equation}
	\tilde{\nu}^{n+1}_x (B)
	:=
	\int_{\Omega^n \times \Omega_X}
	\tilde{\nu}^1_{x_n} (B \rvert_{ \bm{\omega}^{(n)}  })
	\tilde{\nu}^n_{x} (d \bm{\omega}^{(n)} \times dx_n )
	\label{eq:main1}
\end{equation}
for each $1 \leq n < \infty$ and each
$B \in \mathscr{B}^{n+1} \times \mathscr{B}_X$,
where we denote $\prod_{i=1}^n \omega_i$ as $\bm{\omega}^{(n)}$
and 
$
	B \rvert_{ \bm{\omega}^{(n)}  }
	:=
	\Set{
	(\omega_{n+1} , x_{n+1}) \in 
	\Omega \times \Omega_X |
	( \bm{\omega}^{(n)} , \omega_{n+1} , x_{n+1} )
	\in
	B
	}
$
$(1 \leq n \leq \infty).$
From Lemma~\ref{lemm:kernel},
$\tilde{\nu}^n_\cdot (\cdot)$ 
defined by Eq.~\eqref{eq:main1} is a well-defined Markov kernel.
Now we show that
\begin{equation}
	(\mathcal{I}^{ \ast n} \ast F  ) (\cdot)
	=
	\int_{\Omega_X}
	\tilde{\nu}^n_{x} (\cdot)
	F(dx)
	\label{eq:main2}
\end{equation}
for each $n \geq 1$.
If $n=1$, Eq.~\eqref{eq:main2} evident from the definition of $\tilde{\nu}^1_\cdot (\cdot)$.
If Eq.~\eqref{eq:main2} holds for $n \geq 1$,
then for each $B_n \in \mathscr{B}^n , B_1 \in \mathscr{B}$, 
and $B_X \in \mathscr{B}_X$,
we obtain
\begin{align}
	&\int_{\Omega_X}
	\tilde{\nu}^{n+1}_x 
	(B_n \times B_1 \times B^X) F (dx)
	\notag
	\\
	&=
	\int_{\Omega_X}
	\left(
	\int_{\Omega^n \times \Omega_X}
	\tilde{\nu}^1_{x_n} 
	(   
	( B_n \times B_1 \times B_X  ) 
	\rvert_{ \bm{\omega}^{(n)}  })
	\tilde{\nu}^n_{x} (d \bm{\omega}^{(n)} \times dx_n )
	\right)
	F (dx)
	\notag
	\\
	&=
	\int_{\Omega_X}
	\left(
	\int_{\Omega^n \times \Omega_X}
	\chi_{B_n} ( \bm{\omega}^{(n)} )
	\tilde{\nu}^1_{x_n} 
	(   
	 B_1 \times B_X  
	)
	\tilde{\nu}^n_{x} (d \bm{\omega}^{(n)} \times dx_n )
	\right)
	F(dx)
	\notag
	\\
	&=
	\int_{\Omega^n \times \Omega_X}
	\chi_{B_n} ( \bm{\omega}^{(n)} )
	\tilde{\nu}^1_{x_n} 
	(   
	 B_1 \times B_X  
	)
	( \mathcal{I}^{ \ast n} \ast F )
	(d \bm{\omega}^{(n)} \times dx_n )
	\label{eq:toch1}
	\\
	&=
	(\mathcal{I}^{\ast n})_{B_n}
	\left(
	\int_{\Omega_X}
	\tilde{\nu}^1_{x_n} 
	(   
	 B_1 \times B_X  
	)
	F (dx_n)
	\right)
	\label{eq:usedlemm2}
	\\
	&=
	( 
	\mathcal{I}^{\ast (n+1)}
	\ast 
	F
	)
	(
	B_n \times B_1 \times B_X
	),
	\notag
\end{align}
where in deriving Eqs.~\eqref{eq:usedlemm2} and \eqref{eq:toch1}
we have used Lemma~\ref{lemm:koukan} and an equality
\begin{align}
	&\int_{\Omega_X}
	\left(
	\int_{\Omega^n \times \Omega_X}
	f (\bm{\omega}^{(n)}  , x_n )
	\tilde{\nu}^n_{x} (d \bm{\omega}^{(n)} \times dx_n )
	\right)
	F(dx)
	\notag 
	\\
	&=
	\int_{\Omega^n \times \Omega_X}
	f (\bm{\omega}^{(n)}  , x_n )
	( \mathcal{I}^{ \ast n} \ast F )
	(d \bm{\omega}^{(n)} \times dx_n )
	\label{eq:theolemm}
\end{align}
valid for any $\mathscr{B}^n \times \mathscr{B}_X$-measurable bounded function $f$.
Equation~\eqref{eq:theolemm} holds when $f$ is a simple function from the assumption of the induction,
and for general $f$, we can prove the equation by taking a monotone sequence of 
simple functions converging pointwise to $f$.
Thus we have shown Eq.~\eqref{eq:main2}

Now we define a sequence of Markov kernels $\{  \nu^n_\cdot (\cdot) \}$ by
\[
	\nu^n_x (B) :=
	\tilde{\nu}^n_x (B \times \Omega_X)
	\quad
	(x \in \Omega_X, B \in \mathscr{B}^n).
\]
From the definition of $\tilde{\nu}^n_\cdot (\cdot)$,
for each $B \in \mathscr{B}^n$
we have
\begin{align*}
	\nu^{n+1}_x (B \times \Omega)
	&=
	\int_{\Omega^n \times \Omega_X}
	\tilde{\nu}^1_{x_n}
	(
	(B\times \Omega \times \Omega_X)
	\rvert_{\bm{\omega}^{(n)}  }
	)
	\tilde{\nu}^n_x 
	(d \bm{\omega}^{(n)} \times dx_n)
	\\
	&=
	\int_{\Omega^n \times \Omega_X}
	\chi_B(\bm{\omega}^{(n)})
	\tilde{\nu}^n_x (d \bm{\omega}^{(n)} \times dx_n)
	\\
	&=
	\tilde{\nu}^n_x (B \times \Omega_X)
	=
	\nu^n_x (B).
\end{align*}
Thus, from Kolmogorov extension theorem, 
there exists a probability measure $\nu^\infty_x (\cdot)$
with the outcome space $ ( \Omega^\infty , \mathscr{B}^\infty )$
such that
\begin{equation}
	\nu^{\infty}_x (B \times \Omega^\infty)
	=
	\nu^n_x (B)
	\quad
	(n \geq 1, B \in \mathscr{B}^n)
	\label{eq:kernelconsistent}
\end{equation}
for each $x \in \Omega_X$.
To show the $\mathscr{B}_X$-measurability of $\nu^\infty_\cdot (B)$ 
$(B \in \mathscr{B}^\infty)$,
define a class $\mathscr{D}$ of subsets of $\Omega^\infty$ by
\[
	\mathscr{D}
	:=
	\Set{
	B \in \mathscr{B}^\infty
	|
	\text{
	$\nu_\cdot^\infty (B)$ is $\mathscr{B}_X$-measurable
	}
	}.
\]
Then $\mathscr{D}$ is a Dynkin class and, from Eq.~\eqref{eq:kernelconsistent},
$\mathscr{D}$ contains the class
\[
	\Set{
	B \times \Omega^\infty
	|
	B \in \mathscr{B}^n \, (1 \leq n < \infty)
	},
\]
which generates $\mathscr{B}^\infty .$
Therefore the Dynkin\rq{}s theorem assures that $\tilde{\nu}^\infty_\cdot (\cdot )$
is a Markov kernel.
Thus we can define a POVM $\tilde{E}_\infty $
by
\begin{equation}
	\tilde{E}_\infty (B)
	:=
	\int_{\Omega_X}
	\nu^\infty_x (B)
	F (dx)
	\quad
	(B \in  \mathscr{B}^\infty),
	\notag
\end{equation}
which satisfies $\tilde{E}_\infty \preceq F$.
Then for each $B \in \mathscr{B}^n$, we have
\begin{align*}
	\tilde{E}_\infty (B \times \Omega^\infty )
	&=
	\int_{\Omega_X}
	\tilde{\nu}^n_x (B \times \Omega_X)
	F(dx)
	=
	(\mathcal{I}^{\ast n } \ast F ) (B \times \Omega_X)
	=
	E_\infty (B \times \Omega^\infty),
\end{align*}
where we have used Eq.~\eqref{eq:main2} in the second equality.
This implies that
$ E_\infty  = \tilde{E}_\infty \preceq F $,
which completes the proof.
\end{proof}
The part of the result of Theorem~\ref{theo:minimal} 
(conservation of $E_\infty$ by $\mathcal{I}$)
was first obtained in the PhD thesis by the author~\cite{kuramochiphd2014}.

%%%%%%%%%%%%section photon counting and quantum counter
\section{Examples of the infinite composition: photon counting and quantum counter measurements}
\label{sec:4}
In this section, we consider typical examples of standard Borel CP instruments,
namely photon counting~\cite{%
doi10.1080/713820643,0954-8998-1-2-005,PhysRevA.41.4127,PhysRevA.91.032110} 
and quantum counter~\cite{%
PhysRevLett.68.3424,%
PhysRevA.53.3808,PhysRevA.91.032110}
instruments and evaluate the infinite compositions of them.

Let the system Hilbert space $\mathcal{H}$ correspond to a single-mode photon field,
and have a complete orthonormal system $\{  \ket{n} \}_{n \in \mathbb{N}}$
called photon number eigenstates.
Here the set of natural numbers $\mathbb{N}$ contains $0$.
We denote the power set of $\mathbb{N}$ as $2^\mathbb{N}$,
and the countable product space of $(\mathbb{N} , 2^\mathbb{N})$
as $(\mathbb{N}^\infty , \mathscr{B} (\mathbb{N}^\infty)).$

The photon counting and quantum counter instruments for a finite time interval $t > 0$
are discrete and pure CP instruments
with a outcome space $(\mathbb{N} , 2^\mathbb{N})$ defined by
\begin{gather}
	\mathcal{I}^{\text{pc}}_{B } (b)
	:=
	\sum_{m \in B}
	{ M^{\text{pc}}_m }^\ast %(t)
	b
	M^{\text{pc}}_m %(t)
	,
	\notag%\label{eq:pcdef1}
	\\
	M^{\text{pc}}_m %(t)
	:=
	\sum_{n =0}^\infty
	\sqrt{
	p^{\text{pc}}(m+n | n )
	}
	\ket{ n } 
	\bra{n+m} ,
	\notag %\label{eq:pcdef2}
	\\
	p^{\text{pc}}(m| n )
	:=
	\binom{n}{m}
	(1 - e^{-\lambda t})^m e^{\lambda t (n-m)},
	\notag%\label{eq:pcdef3}
\end{gather}
for the photon counting 
instrument~\cite{doi10.1080/713820643,0954-8998-1-2-005,PhysRevA.41.3891,PhysRevA.91.032110}, 
and 
\begin{gather}
	\mathcal{I}^{\text{qc}}_{B } (b)
	:=
	\sum_{m \in B}
	{ M^{\text{qc}}_m }^\ast %(t)
	b
	M^{\text{qc}}_m %(t)
	,
	\notag%\label{eq:qcdef1}
	\\
	M^{\text{qc}}_m  % (t)
	:=
	\sum_{n =0}^\infty
	\sqrt{
	p^{\text{qc}}(m | n )
	}
	\ket{ n+m} 
	\bra{n} ,
	\notag%\label{eq:qcdef2}
	\\
	p^{\text{qc}}(m| n )
	:=
	\binom{n + m}{m}
	(e^{\lambda t} -1 )^m e^{ - \lambda t (n  +  m + 1)},
	\notag%\label{eq:qcdef3}
\end{gather}
for the quantum counter instrument~\cite{PhysRevA.53.3808,PhysRevA.91.032110}.
Here $\lambda$ is a positive constant corresponding to the  coupling strength between
the detector and the photon field.
The infinite composition of $\mathcal{I}^{\text{pc,qc}}$ is a POVM 
$E^{\text{pc,qc}}_\infty$
with the infinite product outcome space 
$(\mathbb{N}^\infty , \mathscr{B} (\mathbb{N}^\infty))$.
Abusing the notation, $\mathcal{I}^{\mathrm{pc},\mathrm{qc}}_{\{  m \}} (\cdot)$
is denoted as $\mathcal{I}^{\mathrm{pc},\mathrm{qc}}_{  m } (\cdot) .$

We define the photon number observable $E^N (\cdot)$
by
\[
	E^N(B)
	:=
	\sum_{n \in B} \ket{n} \bra{n}
	\quad
	(B \in 2^\mathbb{N}),
\]
and a POVM $E^X$ with its outcome space 
$(\mathbb{R}_+  , \mathscr{B}   (  \mathbb{R}_+ )  )$,
where $\mathbb{R}_+$ is a real half-line $(0, \infty)$
and $\mathscr{B}   (  \mathbb{R}_+ )$ is the Borel $\sigma$-algebra
of $\mathbb{R}_+$,
by
\begin{gather}
	E^X (B)
	=
	\int_B
	F_x dx
	\quad
	(B \in \mathscr{B}   (  \mathbb{R}_+ ) ),
	\label{eq:Xdef}
	\\
	F_x
	=
	\sum_{n \in \mathbb{N}}
	\frac{    
	e^{-x} x^n
	}{ 
	n!
	}
	\ket{n} \bra{n} .
	\notag
\end{gather}
Here, $dx$ in Eq.~\eqref{eq:Xdef} is the ordinary Borel measure on the real line.

The following theorem gives explicit forms of the infinite compositions of $\mathcal{I}^{\text{pc,qc}}$.

\begin{theo}
\label{theo:pcqc}
\begin{enumerate}
\item
$E^{\mathrm{pc}}_\infty \simeq E^N  $.
\item
$E^{\mathrm{qc}}_\infty \simeq E^X  $.
\end{enumerate}
\end{theo} 
\begin{proof}
\begin{enumerate}
\item
%\begin{proof}[Proof of 1.]
%%%%%%%%%%%%%%%%proof photon counting%%%%%%%%%%%%%%%%%
From 
\begin{align}
	\mathcal{I}^{\text{pc}}_m
	(\ket{n_1}  \bra{n_1})
	&=
	p^{\mathrm{pc}} (m|m+n_1)
	\ket{  m+ n_1  }  \bra{  m +  n_1}
	\label{eq:pcn}
	\\
	&=
	\sum_{n \in \mathbb{N}}
	\delta_{n, m+n_1}
	p^{\mathrm{pc}} (m|n)
	\ket{n} \bra{n}
	\notag
\end{align}
and
\begin{align*}
	\sum_{ m , n_1 \in \mathbb{N}}
	\delta_{ n ,m + n_1  }
	\mathcal{I}^{\text{pc}}_m
	(\ket{n_1}  \bra{n_1})
	&=
	\sum_{ m , n_1 \in \mathbb{N}}
	\delta_{ n ,m + n_1  }
	p^{\mathrm{pc}} (m|m+n_1)
	\ket{  m+ n_1  }  \bra{  m +  n_1}
	\\
	&=
	\ket{n} \bra{n},
\end{align*}
we obtain $\mathcal{I}^{\mathrm{pc}} \ast E^N \simeq E^N .$
Thus, from Theorem~\ref{theo:minimal}, $E^{\mathrm{pc}}_\infty \preceq E^N$ holds.

In order to show 
$E^N \preceq   E^{\mathrm{pc}}_\infty $,
let us define $\mathscr{B} (\mathbb{N}^\infty)$-measurable
stochastic variables $M_k (\bm{m}^{(\infty)})   $ and 
$ M_\infty(\bm{m}^{(\infty)} )$ 
by
\begin{gather}
	M_k (\bm{m}^{(\infty)})
	:=
	\sum_{i=1}^k m_k \in \mathbb{N} ,
	\label{eq:mkdef}
	\\
	M_\infty (\bm{m}^{(\infty)})
	:=
	\lim_{k \to \infty}
	M_k (\bm{m}^{(\infty)})
	\in 
	\mathbb{N} \cup \{ \infty \} ,
	\notag
\end{gather}
where $\bm{m}^{(\infty)} := (m_1 ,m_2,\cdots) \in \mathbb{N}^\infty$.
Since $\delta_{m , M_k (\bm{m}^{(\infty)}) } \to \delta_{m , M_\infty (\bm{m}^{(\infty)}) }$
for each $m \in \mathbb{N}$ and each $\bm{m}^{(\infty)} \in \mathbb{N}^\infty$ , 
we have
\begin{align}
	E^{M_\infty} (  m)
	&:=
	\int_{\mathbb{N}^\infty}
	\delta_{m , M_\infty (\bm{m}^{(\infty)}) }
	E^{\mathrm{pc}}_\infty (d \bm{m}^{(\infty)}) 
	\notag
	\\
	&=
	\lim_{k \to \infty}
	\int_{\mathbb{N}^\infty}
	\delta_{m , M_k (\bm{m}^{(\infty)}) }
	E^{\mathrm{pc}}_\infty (d \bm{m}^{(\infty)})
	\notag
	\\
	&=
	\lim_{k \to \infty}
	\sum_{ m_1, \cdots , m_k \in \mathbb{N} }
	\delta_{m_1+ \cdots + m_k , m}
	\mathcal{I}^{\mathrm{pc}}_{m_1}
	\circ 
	\cdots
	\circ
	\mathcal{I}^{\mathrm{pc}}_{m_k}
	(I).
	\label{eq:mklaw}
\end{align}
Here 
$E^{M_\infty} (m) = E^{M_\infty}(\{ m \})$ is the POVM derived by $M_\infty$
and the limit is in the sense of the weak operator topology.
Let us evaluate Eq.~\eqref{eq:mklaw}.
From Eq.~\eqref{eq:pcn},
for each $k\geq 1$ and $(m_1 , \cdots , m_k , n_k) \in \mathbb{N}^{k+1}$
we have
\begin{align*}
	&\mathcal{I}^{\mathrm{pc}}_{m_1}
	\circ \cdots \circ
	\mathcal{I}^{\mathrm{pc}}_{m_k}
	(\ket{n_k} \bra{n_k})
	\\
	&=
	\left(
	\prod_{i=1}^k
	p^{\mathrm{pc}}
	(m_i | m_i + \cdots +m_k + n_k)
	\right)
	\ket{m_1 + \cdots + m_k + n_k}
	\bra{m_1 + \cdots + m_k + n_k}
	,
\end{align*}
and thus
\begin{align*}
	&\mathcal{I}^{\mathrm{pc}}_{m_1}
	\circ \cdots \circ
	\mathcal{I}^{\mathrm{pc}}_{m_k}
	(I)
	\\
	&=
	\sum_{m_1 , \cdots , m_k , n \in \mathbb{N}}
%	\left(
%	\prod_{i=1}^k
%	p^{\mathrm{pc}}
%	(m_i | n - m_1 - \cdots - m_{i-1})
%	\right)
	p^{\mathrm{pc}}_k (m_1 , \cdots , m_k |n)
	\ket{n}
	\bra{n},
\end{align*}
where
\begin{equation}
	p^{\mathrm{pc}}_k (m_1 , \cdots , m_k |n)
	:=
	\prod_{i=1}^k
	p^{\mathrm{pc}}
	(m_i | n - m_1 - \cdots - m_{i-1})
	.
	\label{eq:pccond}
\end{equation}
%is a Markov \textit{matrix}.
%Since $p^{\mathrm{pc}}_k (\cdot | n)$ satisfies the Kolmogorov consistency condition for each 
%$n \in \mathbb{N}$,
%we denote the 
By performing some calculations,
the distribution of $m_1 + \cdots + m_k$ for the conditional distribution~\eqref{eq:pccond}
is evaluated to be
\begin{align*}
	p^{\mathrm{pc}}_k (m|n)
	&:=
	\sum_{m_1 , \cdots , m_k \in \mathbb{N}}
	\delta_{ m_1 + \cdots + m_k ,m}
	p^{\mathrm{pc}}_k (m_1 , \cdots , m_k |n)
	\\
	&=
	\binom{ n }{ m  } (1 - e^{-\lambda t k })^m (e^{-\lambda t k})^{n-m} 
	\\
	&\to
	\delta_{n, m} 
	\quad (k \to \infty).
\end{align*}
Thus from Eq.~\eqref{eq:mklaw}, we have
\begin{align}
	E^{M_\infty}
	(m)
	=
	\lim_{k\to \infty}
	\sum_{n \in \mathbb{N}}
	p^{\mathrm{pc}}_k (m|n)
	\ket{n} \bra{n}
	=
	\ket{m} \bra{m} ,
	\label{eq:pclast}
\end{align}
which implies that $E^N = E^{M_\infty} \preceq E^{\mathrm{pc}}_\infty$
and we have proved the assertion.
Note that Eq.~\eqref{eq:pclast} indicates that $E^{M_\infty} (\infty) = I - E^{M_\infty} (\mathbb{N}) = O ,$
i.e. $M_k$ is convergent $E^{\mathrm{pc}}_\infty$-almost surely.
%\end{proof}
%
%\begin{proof}[Proof of 2.]
\item
%%%%%%%%%%%%%%%proof quantum counter%%%%%%%%%%%%%%%%%%
From
\begin{gather*}
	\mathcal{I}^{\mathrm{qc}}_m (F_x)
	=
	e^{-\lambda t}
	p^{\mathrm{qc}} (m | e^{-\lambda t}  x) 
	F_{e^{-\lambda t} x },
	\\
	p^{\mathrm{qc}} (m |x)
	:=
	\frac{ [( e^{\lambda t}  -1) x]^m  }{m!}
	\exp [ - [( e^{\lambda t}  -1) x] ] ,
\end{gather*}
we have
\begin{align}
	\mathcal{I}^{\mathrm{qc}}_{m_1} 
	\circ \cdots \circ
	\mathcal{I}^{\mathrm{qc}}_{m_k}
	(F_x)
	&=
	e^{-\lambda t k}
	\left(
	\prod_{i=1}^k
	p^{\mathrm{qc}} (m_i | e^{- \lambda t(k-i+1)}x)
	\right)
	F_{e^{-\lambda t k }x},
	\notag
	\\
	\mathcal{I}^{\mathrm{qc}}_{m_1} 
	\circ \cdots \circ
	\mathcal{I}^{\mathrm{qc}}_{m_k}
	(I)
	&=
	\int_0^\infty
	\mathcal{I}^{\mathrm{qc}}_{m_1} 
	\circ \cdots \circ
	\mathcal{I}^{\mathrm{qc}}_{m_k}
	(F_x) 
	dx
	\notag
	\\
	&=
	\int_0^\infty
	\left(
	\prod_{i=1}^k
	p^{\mathrm{qc}} (m_i | e^{ \lambda t ( i - 1 )}x)
	\right)
	E^X (dx).
	\label{eq:qc1}
\end{align}
Let $\nu_x^\infty (\cdot)$ be the product measure of $p^{\mathrm{qc}} (\cdot | e^{ \lambda t ( i - 1 )}x)$
with respect to $i \geq 1$ with its outcome space 
$(\mathbb{N}^\infty , \mathscr{B}(\mathbb{N}^\infty ))$.
Then Eq.~\eqref{eq:qc1} implies that
\begin{align}
	E^{\mathrm{qc}}_\infty (\cdot)
	=
	\int_0^\infty \nu_x^\infty (\cdot) E^X (dx),
	\label{eq:qc2}
\end{align}
and we have shown $E^{\mathrm{qc}}_\infty \preceq E^X .$
To show 
$
 E^X 
\preceq
E^{\mathrm{qc}}_\infty 
$,
let us define $\mathscr{B}(\mathbb{N}^\infty)$-measurable stochastic non-negative variables 
$X_k (\bm{m}^{(\infty)})$ by
\begin{align}
	X_k (\bm{m}^{(\infty)})
	:=
	e^{-\lambda t k}
	\sum_{i=1}^k m_i.
	\label{eq:xkdef}
\end{align}
If we denote the expectation with respect to $\nu_x^\infty (\cdot)$
as
$\mathbb{E}_x [ \cdot ]$,
we have
\begin{gather*}
	\mathbb{E}_x
	[ X_k ] 
	=
	(1 - e^{-\lambda t k}) x,
	\\
	\mathbb{E}_x
	[  ( X_k - \mathbb{E}_x
	[ X_k ])^2   ]
	=
	e^{-\lambda t k }
	(1 - e^{-\lambda t k } ) x,
\end{gather*}
where we have omitted the dependence of $\bm{m}^{(\infty)} \in \mathbb{N}^\infty.$
Thus from Chebyshev's inequality
we obtain
\begin{align*}
	\nu_x (\Set{ |X_k -x| > e^{-\lambda t k /4} }  )
	&\leq
	e^{\lambda t k /2}
	\mathbb{E}_x [ |X_k -x|^2  ]
	\leq 
	C_x e^{-\lambda t k /2},
\end{align*}
where $C_x$ is a some positive constant independent of $k.$
Then Borel-Cantelli lemma assures that $X_k $ converges to $x$
$\nu_x^\infty$-almost surely.
Therefore, from Eq.~\eqref{eq:qc2},
$X_k$ is convergent $E^{\mathrm{qc}}_\infty$-almost surely.
Then we can define a non-negative stochastic variable 
$X_\infty := \lim_{k \to \infty} X_k$,
which satisfies $X_\infty = x$ 
$\nu^\infty_x$-almost surely,
i.e. $\nu_x ( X_\infty^{-1} (B) )= \chi_B (x)$
for each $B \in \mathscr{B}(\mathbb{R}_+).$
Therefore, for each $B \in \mathscr{B}(\mathbb{R}_+)$, 
we have
\begin{align*}
	\int_{\mathbb{N}^\infty}
	\chi_B (X_\infty (\bm{m}^{(\infty)}))
	E^{ \mathrm{qc} }_\infty 
	( d \bm{m}^{(\infty)})
	&=
	E^{ \mathrm{qc} }_\infty 
	(X^{-1}_\infty (B))
	\\
	&=
	\int_0^\infty
	\nu_x^\infty 
	(X^{-1}_\infty (B)) E^X (dx)
	\\
	&=
	\int_0^\infty
	\chi_B (x)
	E^X (dx)
	=
	E^X (B),
\end{align*}
which implies $E^X \preceq E^{\mathrm{qc}}_\infty.$
Thus we have proved 
$E^{\mathrm{qc}}_\infty \simeq E^X$.
\qedhere
\end{enumerate}
\end{proof}
%%%%%%%%%%%%%%%%%%%%remark on ueda et al.
According to the above proof,
we also obtain the following theorem concerning the convergences of the stochastic variables
$M_k$ and $X_k$.
\begin{theo}
\label{theo:last}
\begin{enumerate}
\item[(i)]
$M_k$ defined by Eq.~\eqref{eq:mkdef}
is convergent $E^{\mathrm{pc}}_\infty$-almost surely
and 
the POVM corresponding to
the distribution of $\lim_{k \to \infty} M_k$
coincides with the photon number observable $E^N$.
\item[(ii)]
$X_k$ defined by Eq.~\eqref{eq:xkdef}
is convergent $E^{\mathrm{qc}}_\infty$-almost surely
and the POVM corresponding to the distribution of 
$\lim_{k \to \infty}  X_k$
coincides with $E^X$.
\end{enumerate}
\end{theo}

We remark that the statement of Theorem~\ref{theo:last}~(ii)
is essentially the same as 
Theorem~4 of Ref.~\onlinecite{PhysRevA.53.3808}
while the proof of Ref.~\onlinecite{PhysRevA.53.3808}
is, rigorously speaking, insufficient due to the following reason.
The authors of Ref.~\onlinecite{PhysRevA.53.3808}
show that the characteristic function of $X_k$ converges to
that of $E^X$ and conclude the assertion of the theorem.
However, as well-known in the measure theoretic probability theory,
the pointwise convergence of the characteristic function,
which is equivalent to 
the convergence in distribution,
does not necessarily imply the almost sure convergence of a stochastic variable.
In this sense, our proof of Theorem~\ref{theo:last}~(ii) complements
the mathematically insufficient discussion
of Ref.~\onlinecite{PhysRevA.53.3808}.

\section{Summary}
\label{sec:5}
In summary,
we have rigorously reformulated the concept of the information conservation condition
in Definition~\ref{defi:cons}
depending on the equivalence relation among POVMs and 
the composition between an instrument and a POVM.
By using quantum Kolmogorov extension theorem,
we have constructed the infinite composition of a given standard Borel CP instrument $\mathcal{I}$.
We have shown that the infinite composition is the least informative standard Borel POVM
that is conserved by $\mathcal{I}$.
We have considered specific examples of CP instruments, namely
photon counting and quantum counter instruments,
and shown that
their infinite compositions are 
equivalent to the photon number observable $E^N$ 
and the POVM $E^X$ given by Eq.~\eqref{eq:Xdef},
respectively.
As a by-product of the proof,
we have found some results on the almost sure convergences of 
the total counting number for the photon counting 
and the properly normalized counting number
for the quantum counter cases, respectively.
The latter result on the convergence in the quantum counter measurement
complements the insufficiency of the proof in the existing work~\cite{PhysRevA.53.3808}
from the standpoint of the rigorous measure theoretic
description of quantum measurements.

\begin{acknowledgments}
The author acknowledges supports
by Japan Society for the
Promotion of Science (KAKENHI Grant No. 269905).
He also would like to thank helpful discussions with
Masahito Ueda (the Univerisity of Tokyo)
and Tomohiro Shitara (the University of Tokyo).
\end{acknowledgments}

%\nocite{*}
%\bibliography{abbr}

\begin{thebibliography}{21}%
\makeatletter
\providecommand \@ifxundefined [1]{%
 \@ifx{#1\undefined}
}%
\providecommand \@ifnum [1]{%
 \ifnum #1\expandafter \@firstoftwo
 \else \expandafter \@secondoftwo
 \fi
}%
\providecommand \@ifx [1]{%
 \ifx #1\expandafter \@firstoftwo
 \else \expandafter \@secondoftwo
 \fi
}%
\providecommand \natexlab [1]{#1}%
\providecommand \enquote  [1]{``#1''}%
\providecommand \bibnamefont  [1]{#1}%
\providecommand \bibfnamefont [1]{#1}%
\providecommand \citenamefont [1]{#1}%
\providecommand \href@noop [0]{\@secondoftwo}%
\providecommand \href [0]{\begingroup \@sanitize@url \@href}%
\providecommand \@href[1]{\@@startlink{#1}\@@href}%
\providecommand \@@href[1]{\endgroup#1\@@endlink}%
\providecommand \@sanitize@url [0]{\catcode `\\12\catcode `\$12\catcode
  `\&12\catcode `\#12\catcode `\^12\catcode `\_12\catcode `\%12\relax}%
\providecommand \@@startlink[1]{}%
\providecommand \@@endlink[0]{}%
\providecommand \url  [0]{\begingroup\@sanitize@url \@url }%
\providecommand \@url [1]{\endgroup\@href {#1}{\urlprefix }}%
\providecommand \urlprefix  [0]{URL }%
\providecommand \Eprint [0]{\href }%
\providecommand \doibase [0]{http://dx.doi.org/}%
\providecommand \selectlanguage [0]{\@gobble}%
\providecommand \bibinfo  [0]{\@secondoftwo}%
\providecommand \bibfield  [0]{\@secondoftwo}%
\providecommand \translation [1]{[#1]}%
\providecommand \BibitemOpen [0]{}%
\providecommand \bibitemStop [0]{}%
\providecommand \bibitemNoStop [0]{.\EOS\space}%
\providecommand \EOS [0]{\spacefactor3000\relax}%
\providecommand \BibitemShut  [1]{\csname bibitem#1\endcsname}%
\let\auto@bib@innerbib\@empty
%</preamble>
\bibitem [{\citenamefont {Srinivas}\ and\ \citenamefont
  {Davies}(1981)}]{doi10.1080/713820643}%
  \BibitemOpen
  \bibfield  {author} {\bibinfo {author} {\bibfnamefont {M.}~\bibnamefont
  {Srinivas}}\ and\ \bibinfo {author} {\bibfnamefont {E.}~\bibnamefont
  {Davies}},\ }\bibfield  {title} {\enquote {\bibinfo {title} {Photon counting
  probabilities in quantum optics},}\ }\href {\doibase 10.1080/713820643}
  {\bibfield  {journal} {\bibinfo  {journal} {Opt. Acta}\ }\textbf {\bibinfo
  {volume} {28}},\ \bibinfo {pages} {981--996} (\bibinfo {year}
  {1981})}\BibitemShut {NoStop}%
\bibitem [{\citenamefont
  {Ozawa}(1984)}]{content/aip/journal/jmp/25/1/10.1063/1.526000}%
  \BibitemOpen
  \bibfield  {author} {\bibinfo {author} {\bibfnamefont {M.}~\bibnamefont
  {Ozawa}},\ }\bibfield  {title} {\enquote {\bibinfo {title} {Quantum measuring
  processes of continuous observables},}\ }\href {\doibase 10.1063/1.526000}
  {\bibfield  {journal} {\bibinfo  {journal} {J. Math. Phys.}\ }\textbf
  {\bibinfo {volume} {25}},\ \bibinfo {pages} {79--87} (\bibinfo {year}
  {1984})}\BibitemShut {NoStop}%
\bibitem [{\citenamefont {Kuramochi}\ and\ \citenamefont
  {Ueda}(2015)}]{PhysRevA.91.032110}%
  \BibitemOpen
  \bibfield  {author} {\bibinfo {author} {\bibfnamefont {Y.}~\bibnamefont
  {Kuramochi}}\ and\ \bibinfo {author} {\bibfnamefont {M.}~\bibnamefont
  {Ueda}},\ }\bibfield  {title} {\enquote {\bibinfo {title} {Classicality
  condition on a system observable in a quantum measurement and a
  relative-entropy conservation law},}\ }\href {\doibase
  10.1103/PhysRevA.91.032110} {\bibfield  {journal} {\bibinfo  {journal} {Phys.
  Rev. A}\ }\textbf {\bibinfo {volume} {91}},\ \bibinfo {pages} {032110}
  (\bibinfo {year} {2015})}\BibitemShut {NoStop}%
\bibitem [{\citenamefont {Dorofeev}\ and\ \citenamefont
  {de~Graaf}(1997)}]{Dorofeev1997349}%
  \BibitemOpen
  \bibfield  {author} {\bibinfo {author} {\bibfnamefont {S.}~\bibnamefont
  {Dorofeev}}\ and\ \bibinfo {author} {\bibfnamefont {J.}~\bibnamefont
  {de~Graaf}},\ }\bibfield  {title} {\enquote {\bibinfo {title} {Some
  maximality results for effect-valued measures},}\ }\href
  {http://dx.doi.org/10.1016/S0019-3577(97)81815-0} {\bibfield  {journal}
  {\bibinfo  {journal} {Indag. Math. (N.S.)}\ }\textbf {\bibinfo {volume}
  {8}},\ \bibinfo {pages} {349 -- 369} (\bibinfo {year} {1997})}\BibitemShut
  {NoStop}%
\bibitem [{\citenamefont {de~Muynck}(2002)}]{de2002foundations}%
  \BibitemOpen
  \bibfield  {author} {\bibinfo {author} {\bibfnamefont {W.~M.}\ \bibnamefont
  {de~Muynck}},\ }\href@noop {} {\emph {\bibinfo {title} {Foundations of
  quantum mechanics, an empiricist approach}}}\ (\bibinfo  {publisher} {Kluwer
  Academic Publishers, Dordrecht, Boston, London},\ \bibinfo {year}
  {2002})\BibitemShut {NoStop}%
\bibitem [{\citenamefont {Heinonen}(2005)}]{Heinonen200577}%
  \BibitemOpen
  \bibfield  {author} {\bibinfo {author} {\bibfnamefont {T.}~\bibnamefont
  {Heinonen}},\ }\bibfield  {title} {\enquote {\bibinfo {title} {Optimal
  measurements in quantum mechanics},}\ }\href
  {http://dx.doi.org/10.1016/j.physleta.2005.08.003} {\bibfield  {journal}
  {\bibinfo  {journal} {Phys. Lett. A}\ }\textbf {\bibinfo {volume} {346}},\
  \bibinfo {pages} {77 -- 86} (\bibinfo {year} {2005})}\BibitemShut {NoStop}%
\bibitem [{\citenamefont {Jen\v{c}ov\'{a}}, \citenamefont {Pulmannov\'{a}},\
  and\ \citenamefont {Vincekov\'{a}}(2008)}]{jencova2008}%
  \BibitemOpen
  \bibfield  {author} {\bibinfo {author} {\bibfnamefont {A.}~\bibnamefont
  {Jen\v{c}ov\'{a}}}, \bibinfo {author} {\bibfnamefont {S.}~\bibnamefont
  {Pulmannov\'{a}}}, \ and\ \bibinfo {author} {\bibfnamefont {E.}~\bibnamefont
  {Vincekov\'{a}}},\ }\bibfield  {title} {\enquote {\bibinfo {title} {Sharp and
  fuzzy observables on effect algebras},}\ }\href {\doibase
  10.1007/s10773-007-9396-0} {\bibfield  {journal} {\bibinfo  {journal} {Int.
  J. Theor. Phys.}\ }\textbf {\bibinfo {volume} {47}},\ \bibinfo {pages}
  {125--148} (\bibinfo {year} {2008})}\BibitemShut {NoStop}%
\bibitem [{\citenamefont {Ueda}, \citenamefont {Imoto},\ and\ \citenamefont
  {Nagaoka}(1996)}]{PhysRevA.53.3808}%
  \BibitemOpen
  \bibfield  {author} {\bibinfo {author} {\bibfnamefont {M.}~\bibnamefont
  {Ueda}}, \bibinfo {author} {\bibfnamefont {N.}~\bibnamefont {Imoto}}, \ and\
  \bibinfo {author} {\bibfnamefont {H.}~\bibnamefont {Nagaoka}},\ }\bibfield
  {title} {\enquote {\bibinfo {title} {Logical reversibility in quantum
  measurement: General theory and specific examples},}\ }\href {\doibase
  10.1103/PhysRevA.53.3808} {\bibfield  {journal} {\bibinfo  {journal} {Phys.
  Rev. A}\ }\textbf {\bibinfo {volume} {53}},\ \bibinfo {pages} {3808--3817}
  (\bibinfo {year} {1996})}\BibitemShut {NoStop}%
\bibitem [{\citenamefont {Ban}(1999)}]{0305-4470-32-9-012}%
  \BibitemOpen
  \bibfield  {author} {\bibinfo {author} {\bibfnamefont {M.}~\bibnamefont
  {Ban}},\ }\bibfield  {title} {\enquote {\bibinfo {title} {State reduction,
  information and entropy in quantum measurement processes},}\ }\href
  {http://stacks.iop.org/0305-4470/32/i=9/a=012} {\bibfield  {journal}
  {\bibinfo  {journal} {J. Phys. A: Math. Gen.}\ }\textbf {\bibinfo {volume}
  {32}},\ \bibinfo {pages} {1643} (\bibinfo {year} {1999})}\BibitemShut
  {NoStop}%
\bibitem [{\citenamefont {Halmos}\ and\ \citenamefont
  {Savage}(1949)}]{halmos1949application}%
  \BibitemOpen
  \bibfield  {author} {\bibinfo {author} {\bibfnamefont {P.~R.}\ \bibnamefont
  {Halmos}}\ and\ \bibinfo {author} {\bibfnamefont {L.~J.}\ \bibnamefont
  {Savage}},\ }\bibfield  {title} {\enquote {\bibinfo {title} {Application of
  the radon-nikodym theorem to the theory of sufficient statistics},}\
  }\href@noop {} {\bibfield  {journal} {\bibinfo  {journal} {Ann. Math.
  Statist.}\ }\textbf {\bibinfo {volume} {20}},\ \bibinfo {pages} {225--241}
  (\bibinfo {year} {1949})}\BibitemShut {NoStop}%
\bibitem [{\citenamefont {Heinosaari}\ and\ \citenamefont
  {Ziman}(2011)}]{heinosaari2011mathematical}%
  \BibitemOpen
  \bibfield  {author} {\bibinfo {author} {\bibfnamefont {T.}~\bibnamefont
  {Heinosaari}}\ and\ \bibinfo {author} {\bibfnamefont {M.}~\bibnamefont
  {Ziman}},\ }\href@noop {} {\emph {\bibinfo {title} {The mathematical language
  of quantum theory: from uncertainty to entanglement}}}\ (\bibinfo
  {publisher} {Cambridge University Press},\ \bibinfo {year}
  {2011})\BibitemShut {NoStop}%
\bibitem [{\citenamefont {Davies}(1976)}]{davies1976quantum}%
  \BibitemOpen
  \bibfield  {author} {\bibinfo {author} {\bibfnamefont {E.~B.}\ \bibnamefont
  {Davies}},\ }\href@noop {} {\emph {\bibinfo {title} {Quantum theory of open
  systems}}}\ (\bibinfo  {publisher} {IMA},\ \bibinfo {year}
  {1976})\BibitemShut {NoStop}%
\bibitem [{\citenamefont {Davies}\ and\ \citenamefont
  {Lewis}(1970)}]{davieslewisBF01647093}%
  \BibitemOpen
  \bibfield  {author} {\bibinfo {author} {\bibfnamefont {E.~B.}\ \bibnamefont
  {Davies}}\ and\ \bibinfo {author} {\bibfnamefont {J.}~\bibnamefont {Lewis}},\
  }\bibfield  {title} {\enquote {\bibinfo {title} {An operational approach to
  quantum probability},}\ }\href {\doibase 10.1007/BF01647093} {\bibfield
  {journal} {\bibinfo  {journal} {Commun. Math. Phys.}\ }\textbf {\bibinfo
  {volume} {17}},\ \bibinfo {pages} {239--260} (\bibinfo {year}
  {1970})}\BibitemShut {NoStop}%
\bibitem [{\citenamefont {Srivastava}(1998)}]{srivastava1998course}%
  \BibitemOpen
  \bibfield  {author} {\bibinfo {author} {\bibfnamefont {S.~M.}\ \bibnamefont
  {Srivastava}},\ }\href@noop {} {\emph {\bibinfo {title} {A course on Borel
  sets}}}\ (\bibinfo  {publisher} {Springer},\ \bibinfo {year}
  {1998})\BibitemShut {NoStop}%
\bibitem [{\citenamefont {Tumulka}(2008)}]{s11005-008-0229-8}%
  \BibitemOpen
  \bibfield  {author} {\bibinfo {author} {\bibfnamefont {R.}~\bibnamefont
  {Tumulka}},\ }\bibfield  {title} {\enquote {\bibinfo {title} {A kolmogorov
  extension theorem for povms},}\ }\href {\doibase 10.1007/s11005-008-0229-8}
  {\bibfield  {journal} {\bibinfo  {journal} {Lett. Math. Phys.}\ }\textbf
  {\bibinfo {volume} {84}},\ \bibinfo {pages} {41--46} (\bibinfo {year}
  {2008})}\BibitemShut {NoStop}%
\bibitem [{\citenamefont {Kolmogorov}(1933)}]{kolmogorov1933}%
  \BibitemOpen
  \bibfield  {author} {\bibinfo {author} {\bibfnamefont {A.~N.}\ \bibnamefont
  {Kolmogorov}},\ }\href@noop {} {\emph {\bibinfo {title} {Grundbegriffe der
  Wahrscheinlichkeitsrechnung}}}\ (\bibinfo  {publisher} {Springer},\ \bibinfo
  {year} {1933})\BibitemShut {NoStop}%
\bibitem [{\citenamefont {Kuramochi}(2015)}]{kuramochiphd2014}%
  \BibitemOpen
  \bibfield  {author} {\bibinfo {author} {\bibfnamefont {Y.}~\bibnamefont
  {Kuramochi}},\ }\emph {\bibinfo {title} {Relative-entropy conservation law in
  quantum measurement and its applications to continuous measurements}},\
  \href@noop {} {Ph.D. thesis},\ \bibinfo  {school} {Department of Physics, the
  University of Tokyo, to be published in 2016 in
  \url{http://repository.dl.itc.u-tokyo.ac.jp/index_e.html}} (\bibinfo {year}
  {2015})\BibitemShut {NoStop}%
\bibitem [{\citenamefont {Ueda}(1989)}]{0954-8998-1-2-005}%
  \BibitemOpen
  \bibfield  {author} {\bibinfo {author} {\bibfnamefont {M.}~\bibnamefont
  {Ueda}},\ }\bibfield  {title} {\enquote {\bibinfo {title}
  {Probability-density-functional description of quantum photodetection
  processes},}\ }\href {http://stacks.iop.org/0954-8998/1/i=2/a=005} {\bibfield
   {journal} {\bibinfo  {journal} {Quantum Opt.: Journal of the European
  Optical Society Part B}\ }\textbf {\bibinfo {volume} {1}},\ \bibinfo {pages}
  {131} (\bibinfo {year} {1989})}\BibitemShut {NoStop}%
\bibitem [{\citenamefont {Imoto}, \citenamefont {Ueda},\ and\ \citenamefont
  {Ogawa}(1990)}]{PhysRevA.41.4127}%
  \BibitemOpen
  \bibfield  {author} {\bibinfo {author} {\bibfnamefont {N.}~\bibnamefont
  {Imoto}}, \bibinfo {author} {\bibfnamefont {M.}~\bibnamefont {Ueda}}, \ and\
  \bibinfo {author} {\bibfnamefont {T.}~\bibnamefont {Ogawa}},\ }\bibfield
  {title} {\enquote {\bibinfo {title} {Microscopic theory of the continuous
  measurement of photon number},}\ }\href {\doibase 10.1103/PhysRevA.41.4127}
  {\bibfield  {journal} {\bibinfo  {journal} {Phys. Rev. A}\ }\textbf {\bibinfo
  {volume} {41}},\ \bibinfo {pages} {4127--4130} (\bibinfo {year}
  {1990})}\BibitemShut {NoStop}%
\bibitem [{\citenamefont {Ueda}\ and\ \citenamefont
  {Kitagawa}(1992)}]{PhysRevLett.68.3424}%
  \BibitemOpen
  \bibfield  {author} {\bibinfo {author} {\bibfnamefont {M.}~\bibnamefont
  {Ueda}}\ and\ \bibinfo {author} {\bibfnamefont {M.}~\bibnamefont
  {Kitagawa}},\ }\bibfield  {title} {\enquote {\bibinfo {title} {Reversibility
  in quantum measurement processes},}\ }\href {\doibase
  10.1103/PhysRevLett.68.3424} {\bibfield  {journal} {\bibinfo  {journal}
  {Phys. Rev. Lett.}\ }\textbf {\bibinfo {volume} {68}},\ \bibinfo {pages}
  {3424--3427} (\bibinfo {year} {1992})}\BibitemShut {NoStop}%
\bibitem [{\citenamefont {Ueda}, \citenamefont {Imoto},\ and\ \citenamefont
  {Ogawa}(1990)}]{PhysRevA.41.3891}%
  \BibitemOpen
  \bibfield  {author} {\bibinfo {author} {\bibfnamefont {M.}~\bibnamefont
  {Ueda}}, \bibinfo {author} {\bibfnamefont {N.}~\bibnamefont {Imoto}}, \ and\
  \bibinfo {author} {\bibfnamefont {T.}~\bibnamefont {Ogawa}},\ }\bibfield
  {title} {\enquote {\bibinfo {title} {Quantum theory for continuous
  photodetection processes},}\ }\href {\doibase 10.1103/PhysRevA.41.3891}
  {\bibfield  {journal} {\bibinfo  {journal} {Phys. Rev. A}\ }\textbf {\bibinfo
  {volume} {41}},\ \bibinfo {pages} {3891--3904} (\bibinfo {year}
  {1990})}\BibitemShut {NoStop}%
\end{thebibliography}

\end{document}